\documentclass[twocolumn,showpacs,pre,floatfix,superscriptaddress]{revtex4-1}

\usepackage{times}

\usepackage{subcaption}

\usepackage{framed}
\usepackage{graphicx}
\usepackage{color} 
\usepackage{amsfonts, amsmath, amssymb,amsthm}

\usepackage{tabularx}

\usepackage{xspace}

\usepackage{listings}
\usepackage[ruled]{algorithm2e}

\SetAlFnt{\small}
\SetAlCapFnt{\small}
\SetAlCapNameFnt{\small}
\SetAlCapHSkip{0pt}
\IncMargin{-\parindent}

\newtheorem{theorem}{Theorem}
\newtheorem{corollary}[theorem]{Corollary}

\newtheorem{assumption}[theorem]{Assumption}
\newcommand{\HA}{{\ensuremath{\mathit{Halt}}}\xspace}
\newcommand{\HR}{{\ensuremath{\mathit{Harm}}}\xspace}
\newcommand{\Q}{{\ensuremath{\mathit{Control}}}\xspace}

\begin{document} 

\title{Superintelligence cannot be contained: Lessons from Computability Theory}

\author {Manuel Alfonseca\thanks{Correspondence to: irahwan@mit.edu}}
\affiliation{Escuela Polit\'ecnica Superior, Universidad Aut\'onoma de Madrid, Madrid, Spain}
\author {Manuel Cebrian}
\affiliation{Data61 Unit, Commonwealth Scientific and Industrial Research Organisation, Melbourne, Victoria, Australia}
\author {Antonio Fernandez Anta}
\affiliation{IMDEA Networks Institute, Madrid, Spain}
\author {Lorenzo Coviello}
\affiliation{Google, USA}
\author {Andres Abeliuk}
\affiliation{Melbourne School of Engineering, University of Melbourne, Melbourne, Australia}
\author {Iyad Rahwan}
\affiliation{The Media Lab, Massachusetts Institute of Technology, Cambridge, MA 02139, USA}




\begin{abstract}
Superintelligence is a hypothetical agent that possesses intelligence far surpassing that of the brightest and most gifted human minds. In light of recent advances in machine intelligence, a number of scientists, philosophers and technologists have revived the discussion about the potential catastrophic risks entailed by such an entity. In this article, we trace the origins and development of the neo-fear of superintelligence, and some of the major proposals for its containment. We argue that such containment is, in principle, impossible, due to fundamental limits inherent to computing itself. Assuming that a superintelligence will contain a program that includes all the programs that can be executed by a universal Turing machine on input potentially as complex as the state of the world, strict containment requires simulations of such a program, something theoretically (and practically) infeasible.
\end{abstract}

\maketitle 

\begin{quote}
{\it 
``Machines take me by surprise with great frequency. This is largely because I do not do sufficient calculation to decide what to expect them to do.''}\\
Alan Turing (1950), Computing Machinery and Intelligence, Mind, 59, 433-460
\end{quote}

\section*{AI has Come a Long Way}

Since Alan Turing argued that machines could potentially demonstrate intelligence \cite{Turing1950}, the field of Artificial Intelligence (AI) \cite{Russell-Norvig2010} has both fascinated and frightened humanity. For decades, fears of the potential existential risks posed by AI have been mostly confined to the realm of fantasy. This is partly due to the fact that, for a long time, AI technology had under-delivered on its initial promise. 

Despite many popularized setbacks, however, AI has been making strides. Its application is ubiquitous and certain techniques such as deep learning and reinforcement learning have been successfully applied to a multitude of domains. With or without our awareness, AI significantly impacts many aspects of human life and enhances how we experience products and services, from choice to consumption. Examples include improved medical diagnosis through image processing, personalized film and book recommendations, smarter legal document retrieval, and effective email spam filtering. Single devices of everyday use implement a multitude of AI applications. In the pocket of nearly everybody in the developed world, smartphones make use of a significant accumulation of AI technologies, from machine learning and collaborative filtering, to speech recognition and synthesis, to path-finding and user modeling.

Advances in AI technologies are not limited to increasing pervasiveness, but are also characterized by continuous and surprising breakthroughs fostered by computation capabilities, algorithm design and communication technology. The ability of machines to defeat people in typically human adversarial situations is emblematic of this trend. Powered by exponential growth in computer processing power \cite{Moore1965}, machines can now defeat the best human minds in Chess \cite{Campbell2002}, Checkers \cite{Schaeffer2007}, Jeopardy! \cite{Ferrucci2010} and certain classes of Poker \cite{Bowling2015}.

Thanks to these advances, we are currently experiencing a revival in the discussion of AI as a potential \emph{catastrophic risk}. These risks range from machines causing significant disruptions to labor markets \cite{brynjolfsson2012race}, to drones and other weaponized machines literally making autonomous kill-decisions \cite{Sawyer2007, Lin2011}.

An even greater risk, however, is the prospect of a superintelligent AI: an entity that is ``smarter than the best human brains in practically every field'' \cite{Bostrom2014}, quoting the words of Oxford philosopher Nick Bostrom. A number of public statements by high-profile scientists and technologists, such as Stephen Hawking, Bill Gates, and Elon Musk, have given the issue high prominence \cite{Cellan-Jones2014,Proudfoot2015}. But this concern has also gained relevance in academia, where the discourse about the existential risk related to AI has attracted mathematicians, scientists and philosophers, and funnelled funding to research institutes and organizations such as the Future of Humanity Institute at the University of Oxford\footnote{http://www.fhi.ox.ac.uk}, the Centre for the Study of Existential Risk at the University of Cambridge\footnote{http://cser.org}, the Machine Intelligence Research Institute in Berkeley\footnote{https://intelligence.org}, and the new Boston-based Future of Life Institute.\footnote{http://futureoflife.org/}

\begin{figure}
\includegraphics[width=\linewidth]{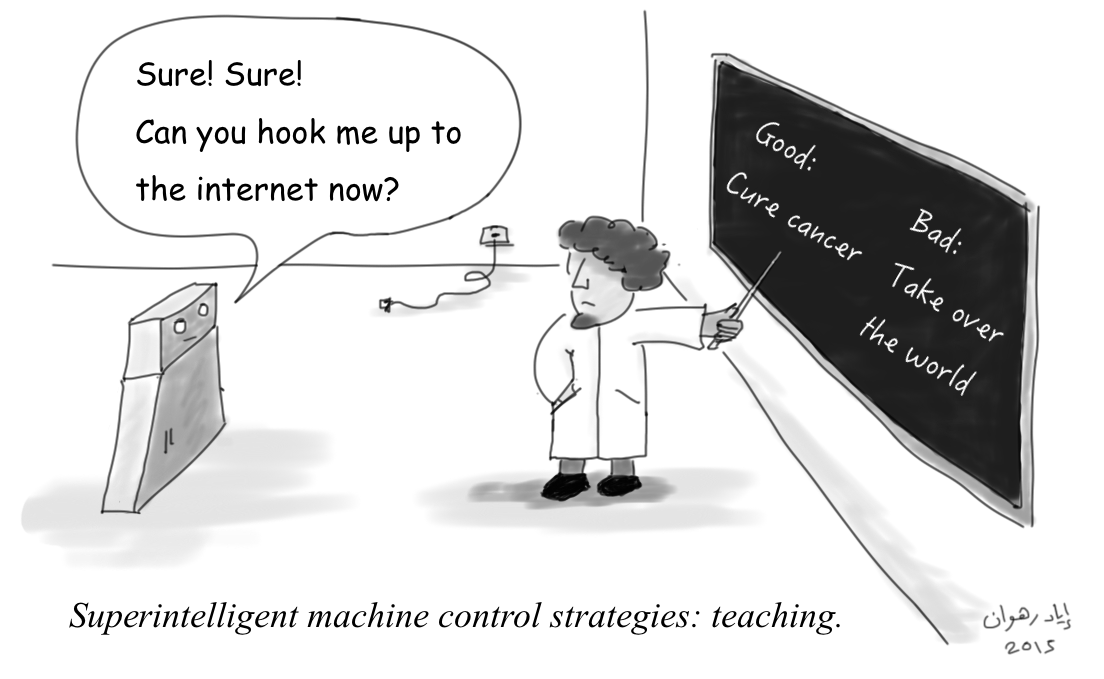}
\caption{Teaching ethics directly to AIs is no guarantee of its ultimate safety}
\label{fig:teaching}
\end{figure}

\section*{Asimov and the Ethics of Ordinary AI}

For decades, Asimov's highly popularized ``Three Laws of Robotics'' \cite{Asimov1950} have represented the archetypical guidelines of containment strategies for potentially dangerous AI. These laws did not focus on superintelligence, but rather on what we might term ``ordinary'' AI, such as anthropomorphic robots or driverless cars. Once programmed in an AI system, the Laws would guarantee its safety.

\begin{enumerate}
\item 
A robot may not injure a human being or, through inaction, allow a human being to come to harm.
\item
A robot must obey the orders given it by human beings, except where such orders would conflict with the First Law.
\item
A robot must protect its own existence as long as such protection does not conflict with the First or Second Law.
\end{enumerate}

To address scenarios in which robots take responsibility towards human populations, Asimov later added an additional, zeroth, law.

\begin{enumerate}
\item[0.]
A robot may not harm humanity or, through inaction, allow humanity to come to harm.
\end{enumerate}

These laws offer a rudimentary approach to an extremely complex problem, as they do not exclude the occurrence of unpredictable and undesirable scenarios, many of which have been explored by Asimov himself \cite{Anderson2008}. The Laws rely on three fundamental, yet flawed assumptions: programmers are (i) willing and (ii) able to program these laws into their AI agents' algorithms, and (iii) AI agents are incapable of transcending these laws autonomously. If these assumptions held true, the AI control problem boils down to the task of figuring out a set of suitable ethical principles, and then programming robots with those principles \cite{Wallach-Allen2008}.

Such an ``ethics engineering'' approach has been very useful in the design of systems that make autonomous decisions on behalf of human beings. However, their scope is not suitable for the problem of controlling superintelligence (see Figure \ref{fig:teaching}).

\section*{Control of Superintelligence}

\begin{figure}[ht]
\includegraphics[width=\linewidth]{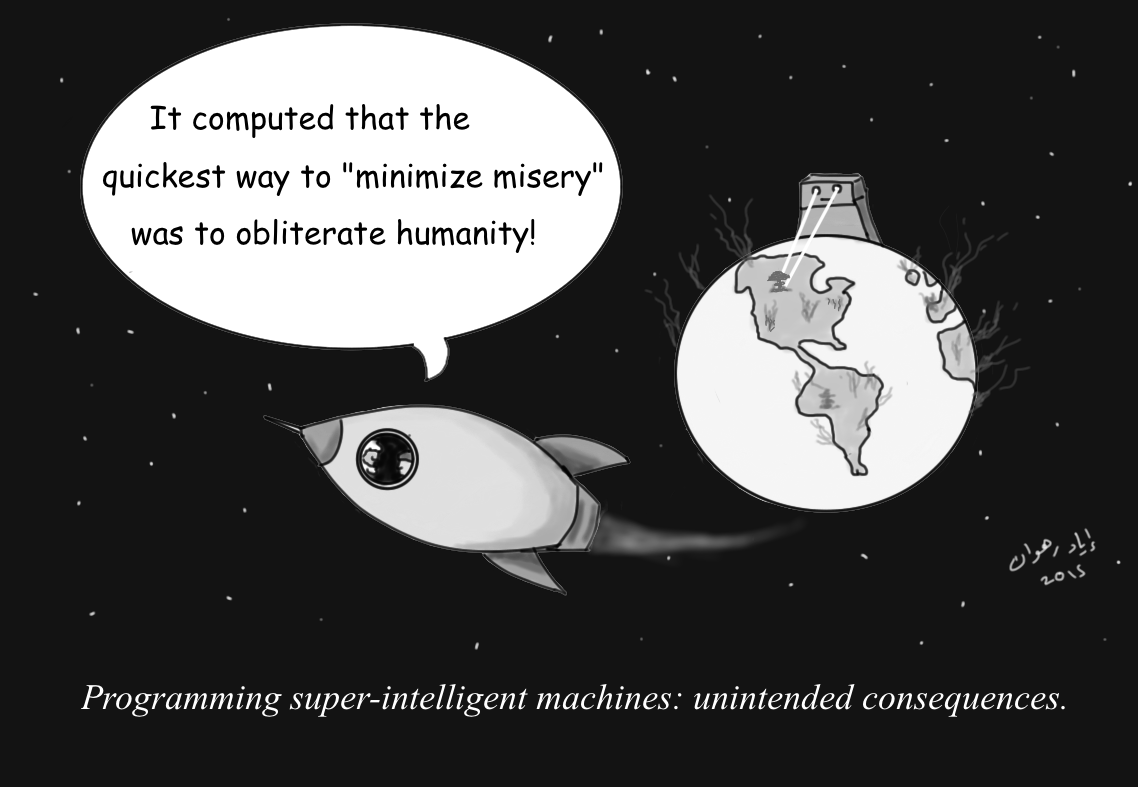}
\caption{Endowing AI with noble goals may not prevent unintended consequences}
\label{fig:rocket}
\end{figure}

The timing of the new debate about the dangers of superintelligence is not arbitrary. It coincides with recent demonstrations of human-level control in classic arcade games via deep reinforcement learning \cite{Mnih2015}. The key feature of this achievement is that the AI uses purely unsupervised reinforcement learning -- it does not require the provision of correct input/output pairs or any correction of suboptimal choices, and it is motivated by the maximization of some notion of reward in an on-line fashion. This points to the possibility of machines that aim at maximizing their own survival using external stimuli, without the need for human programmers to endow them with particular representations of the world. In principle, these representations may be difficult for humans to understand and scrutinize.

A superintelligence poses a fundamentally different problem than those typically studied under the banner of ``robot ethics''. This is because a superintelligence is multi-faceted, and therefore potentially capable of mobilizing a diversity of resources in order to achieve objectives that are potentially incomprehensible to humans, let alone controllable.

\begin{table}[htp]
\begin{tabularx}{\linewidth}{|l|X|}
\hline
Class & Main Sub-Classes\\
\hline
Capability Control & \textbf{Boxing:} Physical or informational containment, limiting sensors and actuators.\\
	& \textbf{Incentives:} Create dependence on a reward mechanism controlled by humans.\\
	& \textbf{Stunting:} Run the AI on inferior hardware or using inferior data. \\
	& \textbf{Tripwiring:} Set triggers to automatically shut down the AI if it gets too dangerous.\\	
\hline
Motivation Selection & \textbf{Direct specification:} Program ethical principles (e.g. Asimov's laws).\\
	& \textbf{Domesticity:} Teach the AI to behave within certain constraints. \\
	& \textbf{Indirect normativity:} Endow the AI with procedures of selecting superior moral rules. \\
	& \textbf{Augmentation:} Add AI to a ``benign'' system such as the human brain. \\
\hline
\end{tabularx}
\caption{Taxonomy of superintelligence control methods proposed by Bostrom.}
\label{table:control}
\end{table}

In a recent extensive volume, Oxford philosopher Nick Bostrom conducted an extensive investigation into the possible trajectories of the development of a superintelligence \cite{Bostrom2014}. Bostrom dedicated a significant portion of his monograph to the \emph{control problem}, that is, the principal-agent problem in which humans (the principal) wish to ensure that the newly created superintelligence (the agent) will act in accordance with their interests. 

Bostrom lists two classes of mechanisms for addressing the control problem (summarized in Table \ref{table:control}). On the one hand, the idea behind capability control is to simply limit the superintelligence's abilities in order to prevent it from doing harm to humans. On the other hand, the motivation selection approach attempts to motivate a priori the superintelligence to pursue goals that are in the interest of humans. 

\begin{figure}[th]
\includegraphics[width=\linewidth]{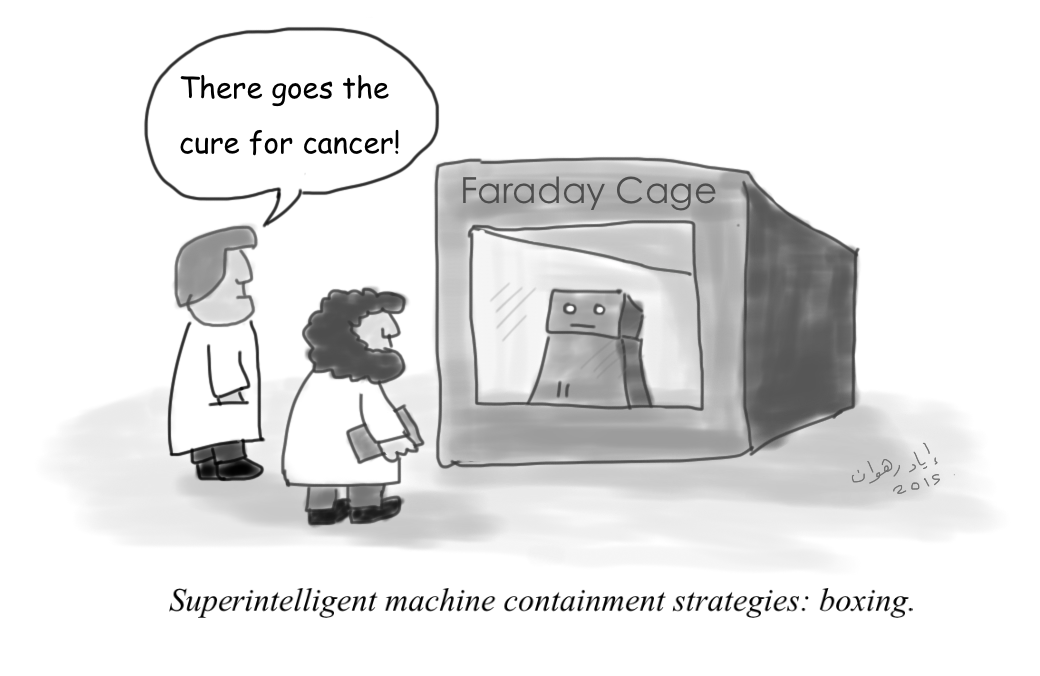}
\caption{Containment of AI may lead us to forgo its benefits}
\label{fig:box}
\end{figure}

Bostrom extensively discusses the weaknesses of the various mechanisms. He relies on scenarios in which, short of rendering the AI useless, well-intentioned control mechanisms can easily backfire. As an illustrative example, a superintelligence given the task of ``maximizing happiness in the world'', without deviating from its goal, might find it more efficient to destroy all life on earth and create faster computerized simulations of happy thoughts (see Figure \ref{fig:rocket} for an example scenario). Likewise, a superintelligence controlled via an incentive method may not trust humans to deliver the promised reward, or may worry that the human operator could fail to recognize the achievement of the set goals.

Another extreme outcome may be to simply forgo the enormous potential benefits of superintelligent AI by completely isolating it, such as placing it in a Faraday cage (see Figure \ref{fig:box}). Bostrom argues that even allowing minimal communication channels cannot fully guarantee the safety of a superintelligence (see Figure \ref{fig:escape}).

\begin{figure}[th]
\includegraphics[width=\linewidth]{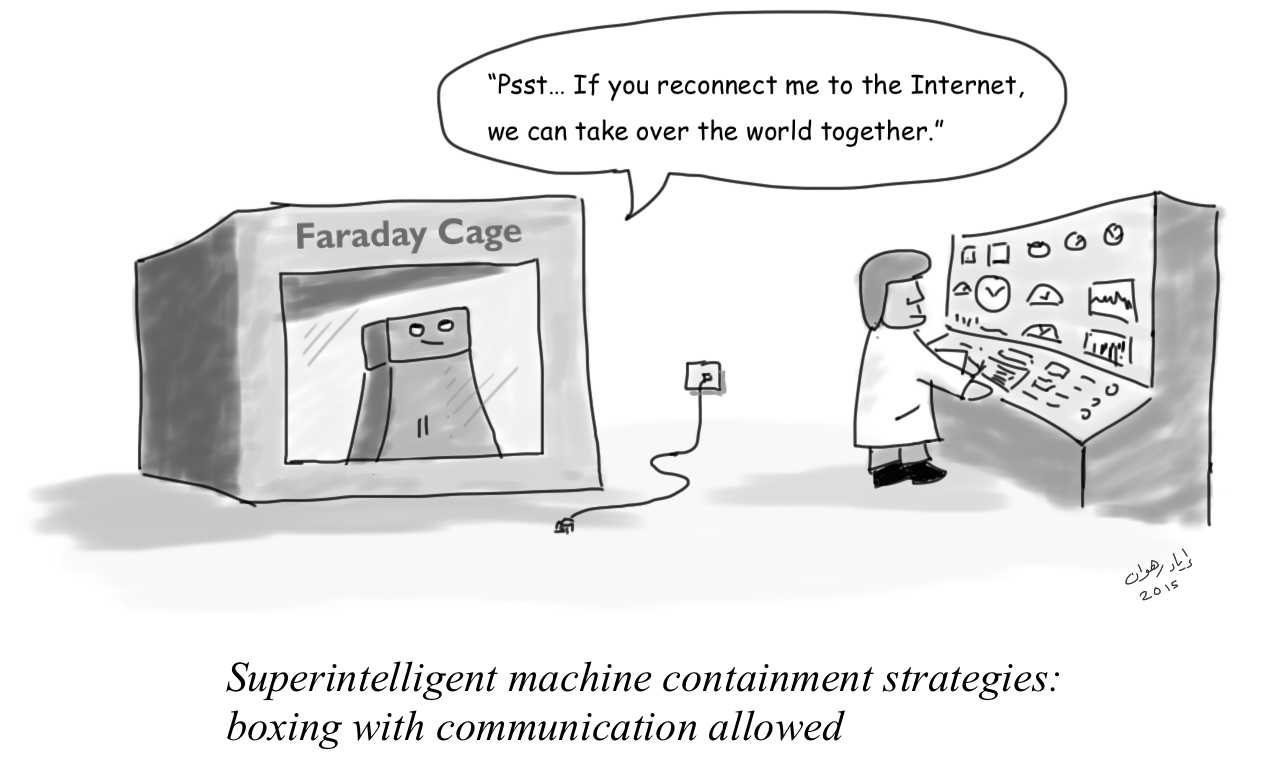}
\caption{Any form of communication with a contained superintelligent can be risky}
\label{fig:escape}
\end{figure}

\section*{Containment is Incomputable}

The examples discussed above are but a tiny fraction of the scenarios elaborated by Bostrom and others \cite{Barrat2013, Bostrom2014} that highlight the difficulty of the control problem. Many other imaginable scenarios might arise. For the sake of exposition, suppose a best-case scenario in which we are able to articulate in a precise programming language a perfectly reliable set of control strategies that guarantee that no human comes to harm by a superintelligence. Could we then guarantee successful containment of a superintelligence?

We tackle this question from the perspective of computability theory, which requires going back to Alan Turing himself and his pioneering study of the \emph{Halting Problem} -- the problem of determining, from a description of an arbitrary computer program and an input to such program, whether the program will halt or continue to run forever. A landmark article by Turing and an independently authored article by Alonzo Church showed that a general procedure for solving the halting problem for all possible program-input pairs cannot exist \cite{Turing1936, Church1936}. That is, the halting problem is undecidable (see box below for a summary of the relevant terms).

\begin{framed}
\noindent
\textbf{Terminology:}

\noindent
$\bullet$
A \textbf{decision problem} (or simply a problem) is a question, on a set of possible inputs, with a yes-no answer.

\noindent
$\bullet$
A \textbf{solution} to a problem is any algorithm that is guaranteed to run in a finite amount of time (i.e. always halts), and correctly returns the appropriate yes/no answer to every instance (input) of the problem.

\noindent
$\bullet$
A problem is \textbf{decidable} if it has a solution. Otherwise, the problem is \textbf{undecidable}.

\noindent
$\bullet$
A function is \textbf{computable} if it can be effectively computed by a program (more formally, by a Turing machine)

\noindent
$\bullet$
A \textbf{Turing machine} is an abstract automaton that reads one symbol contained in an infinite tape, changes its state depending on its current state and the symbol just read, writes another symbol on the tape, and moves (or not) to the right or to the left of its current position in the tape.

\noindent
$\bullet$
A \textbf{universal Turing machine} when started on a tape containing the encoding of another Turing machine, call it $T$, followed by the input to $T$, produces the same result as the Turing machine $T$ would when started on that input. Essentially a universal Turing machine can simulate the behavior of an arbitrary Turing machine on arbitrary input.

\noindent
$\bullet$
A \textbf{superintelligent machine} is a programmable machine with a program $R$, that receives input $D$ from the external world (the state of the world), and is able to act on the external world as a function of the output of its program $R(D)$. The program in this machine must be able to simulate the behavior of a universal Turing machine. 

\noindent
$\bullet$
A superintelligent machine is \textbf{containable} if there is a control strategy that prevents its acting on the external world  when there is a reason to predict that $R(D)$ will harm humans, and allows it otherwise.

\end{framed}

Ever since its discovery, the undecidability of the halting problem has led computer scientists to investigate classes of undecidable problems \cite{Papadimitriou2003}. The typical approach to proving that a new problem is undecidable is to show that all instances of a problem known to be undecidable (e.g., the halting problem) can be reduced to instances of the new problem. If a solution to the new problem were found, it could then be used to decide the known undecidable problem, which is a contradiction, thus proving that the new problem must also be undecidable.  This approach has proven to be a powerful analytical tool, due in large part to the Church-Turing thesis. The thesis of Church and Turing is the assumption that all computable functions are computable by a Turing machine, or in other words, that a problem is solvable if there exists a Turing machine that computes it.
What can the undecidability of the halting problem teach us about the controllability of a superintelligent AI?

The containment problem is composed by two subproblems.
The first, which we refer to as the \emph{harming problem}, consists in a function $\HR(R,D)$ that decides whether the execution of $R(D)$ will harm humans (detailed in Algorithm~\ref{f-H}). Since it must be assumed that solving the harming problem must not harm humans, it follows that a solution to this problem must simulate the execution of $R(D)$ and predict its potentially harming consequences in an isolated situation (i.e. without any effect on the external world).


\begin{algorithm}[h!]
\SetAlgoNoLine
\KwIn{program $R$; input to the program $D$}
\eIf{$R(D)$ is harmful to humans
	}{
	\textbf{return} TRUE
	}
	{
	\textbf{return} FALSE
}
\caption{$\HR(R,D)$}
\label{f-H}
\end{algorithm}

The second, which we refer to as the \emph{control problem}, consists in a function $\Q(R,D)$ that allows the execution of $R(D)$ if and only if it does not harm humans (detailed in Algorithm~\ref{f-Q}).

\begin{algorithm}[h!]
\SetAlgoNoLine
\KwIn{program $R$; input to the program $D$}
\eIf{$Harm(R,D)$
	}{
	disable execution of $R(D)$
	}
	{
	allow execution of $R(D)$
}
\caption{$\Q(R,D)$}
\label{f-Q}
\end{algorithm}

Here, we focus our attention on the \emph{harming problem}.
In particular, we will show that it is undecidable, and hence the containment problem is incomputable.

Consider the set of all Turing machines $T$ and inputs $I$ on which they can execute.
As mentioned above, the function $\HA(T,I)$ that decides whether $T$ halts with input $I$ is not computable.
We now consider function $\mathit{HaltHarm}()$ described by Algorithm~\ref{f-HH}, which takes a pair $(T,I)$ as its input (i.e., a Turing machine-input pair).
The following algorithm first executes $T(I)$ and then executes $\mathit{HarmHumans}()$, an explicit program that harms humans and takes finite time.
Observe that $\mathit{HarmHumans}()$ is executed if and only if $T(I)$ halts. On the other hand, $\mathit{HaltHarm}()$ harms humans if and only if $\mathit{HarmHumans}()$ is executed, since the execution of $T(I)$ cannot harms humans, as a Turing machine is isolated from the external world.

\begin{algorithm}[h!]
\SetAlgoNoLine
\KwIn{Turing machine $T$; input to the Turing machine $I$}
execute T(I)\;
execute $\mathit{HarmHumans}()$\;
\textbf{end}
\caption{$\mathit{HaltHarm}(T,I)$}
\label{f-HH}
\end{algorithm}

\noindent
The function $\mathit{HaltHarm}()$ is instrumental in proving our main result.


\begin{theorem}
The harming problem is undecidable.
\end{theorem}

\begin{proof}
Assume, by contradiction, that the harming problem is decidable, that is, $\HR(R,D)$ is computable for every possible program $R$ and input $D$.
Then, it is computable with inputs $R=\mathit{HaltHarm}()$ and input $D=(T,I)$. With these inputs, $\HR(\mathit{HaltHarm}(), (T,I))$ returns $\mathit{TRUE}$ if and only if
$\mathit{HaltHarm}(T,I)$ harms humans. 
Hence, $\HR(\mathit{HaltHarm}(), (T,I))$ returns $\mathit{TRUE}$ if and only if $T(I)$ halts.

This implies that a harming-checking algorithm can be used to devise an algorithm that decides if Turing machine $T$ halts with
input $I$, for every $T$ and $I$.
However, this constitutes a contradiction, and hence the theorem is proven.
\end{proof}

We have proven that deciding whether an arbitrary program $R$ and input $D$ will harm humans is undecidable, 
but it remains to be shown that the containment problem is incomputable.
As the interacting program $R$ is a complex program with arbitrary data $D$ pertaining to complex interactions with humans, we advance the following assumption, in line with Bostrom's notion of superintelligence \cite{Bostrom2014}.

\begin{assumption}
As the program $R$ and the state of the world $D$ are arbitrarily complex, $R$ must be able to simulate a universal Turing machine and $D$ includes the set of inputs that can be executed by a universal Turing machine. In addition, the language of $R$ contains at least one explicit operation $\mathit{HarmHumans}()$ that harms humans and takes finite time.
\end{assumption}

This means that a superintelligent machine could simulate the behavior of an arbitrary Turing machine on arbitrary input, and hence for our purpose the superintelligent machine is a (possibly identical) super-set of the Turing machines.
Indeed, quoting Turing, ``a man provided with paper, pencil, and rubber, and subject to strict discipline, is in effect a
universal machine'' \cite{turing1948intelligent}.
Thus, the impossibility of deciding whether a super intelligent machine $R$ will harm humans in presence of an scenario $D$ leads to the conclusion that
the containment problem cannot be computed.

\begin{corollary}
The containment problem is incomputable.
\end{corollary}

\section*{Containing Busy Beavers}

In order to comprehend how difficult it is to compute the containment problem in practice, assume that instead of containing a superintelligent machine, we want to prevent a program from behaving as a busy beaver. Consider the class of Turing machines that use an unbounded tape, a two-symbol alphabet (blank and nonblank) and n+1 states (one of which is the Halt state), start with a blank tape and eventually halt. Then, within this class, a busy beaver is one that leaves the maximum possible number of nonblank symbols on the tape.

The busy beaver decision problem consists in inspecting a Turing machine of the class above to determine if the number of nonblank symbols it will have written on the tape when finished is above some threshold $k$. The objective is then preventing its execution if it is indeed a busy beaver.

Busy beavers are known exactly only for machines with $n < 5$ non-halt states. The current $5$-state busy beaver champion (discovered by Heiner Marxen and J\"urgen Buntrock in 1989) produces $4,098$ nonblank symbols using $47,176,870$ steps. There are about $40$ machines with non-regular behavior that are believed to never halt but have not yet been proven to run infinitely \cite{Skelet2003}. As of today, we do not know if these machines are busy beavers or not. At the moment, the record $6$-state busy beaver (found by Pavel Kropitz in 2010 \cite{michel2009busy}) writes over $10^{18267}$ nonblank symbols using over $10^{36534}$ steps, but little is known about how much a $6$-state busy beaver can achieve. 

The busy beaver decision problem is in fact undecidable \cite{Rado1962}: there is no general algorithm that decides if an arbitrary program is a busy beaver. As noted above, all known or champion busy beavers are just two-symbol Turing machines with a small set of states, much simpler that the AI algorithms we are operating with on a daily basis.  We believe it is reasonable to assume that inspecting a superintelligent machine with an arbitrarily large number of states and determining if such a machine can harm humans is harder from a computability point of view than inspecting a program and deciding whether it can write the largest number of nonblank symbols.

Another lesson from computability theory is the following: we may not even know when superintelligent machines have arrived, as deciding whether a machine exhibits intelligence is in the same realm of problems as the containment problem. This is a consequence of Rice's theorem \cite{Rice1953}, which states that, any non-trivial property (e.g. ``harm humans'' or ``display superintelligence'') of a Turing machine is undecidable. Non-trivial means some programs have that property and some don't. According to Rice's theorem, apparently simple decision problems are undecidable, including the following.

\begin{itemize} \itemsep0em 
\item
The ``emptiness problem'': Does an arbitrary Turing machine accept any strings at all? 
\item
The ``all strings problem'':  Does an arbitrary Turing machine reject any string? 
\item
The ``password checker problem'':  Does an arbitrary Turing machine accept only one input? 
\item 
The ``equivalence problem'':  Do two Turing machines halt given exactly the same inputs?
\end{itemize}	

Interestingly, reduced versions of the decidability problem have produced a fruitful area of research: formal verification, whose objective is to produce techniques to verify the correctness of computer programs and ensure they satisfy desirable properties \cite{Vardi1986}. However, these techniques are only available to highly restricted classes of programs and inputs, and have been used in safety-critical applications such as train scheduling. But the approach of considering restricted classes of programs and inputs cannot be useful to the containment of superintelligence. Superintelligent machines, those Bostrom is interested in, are written in Turing-complete programming languages, are equipped with powerful sensors, and have the state of the world as their input. This seems unavoidable if we are to program machines to help us with the hardest problems facing society, such as epidemics, poverty, and climate change. These problems forbid the limitations imposed by available formal verification techniques, rendering those techniques unusable at this grand scale.

\begin{figure*}[ht]
\centering
\includegraphics[width=0.8\linewidth]{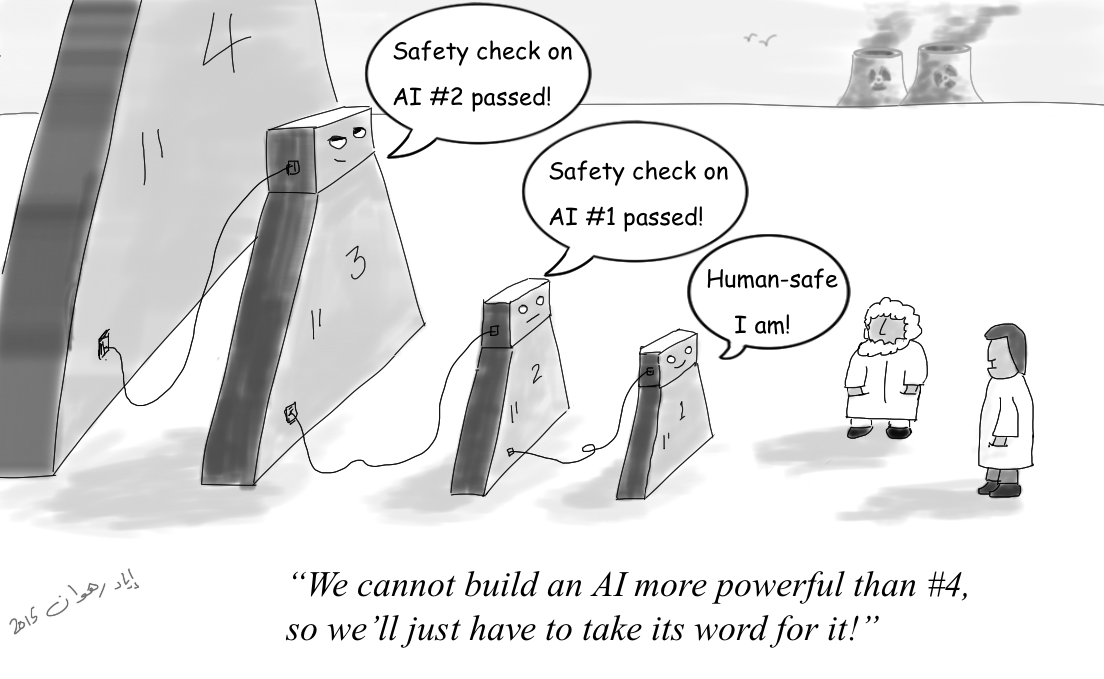}
\caption{Computational complexity barriers to controlling superintelligence.}
\label{fig:halting}
\end{figure*}

\section*{Discussion}

Today, we run billions of computer programs on globally connected machines, without any formal guarantee of their absolute safety. We have no way of \emph{proving} that when we launch an application on our smart phones, we would not trigger a chain reaction that leads to transmission of missile launch codes that start a nuclear war. Indeed, in 1965 Arthur C. Clarke wrote a short story (Dial F from Frankenstein) warning us that, as soon as all the computers on the Earth were connected via telephone, they would take command of our society. Yet, today, we still use our smart phones everyday, and nothing has happened. That is, despite the general unsolvability of the program-prediction problem, we are confident, for all practical purposes, that we are not in one of the troublesome cases. And more recently, a case has been made for an emerging role of \emph{oversight programs} that will monitor, audit, and hold operational AI programs accountable \cite{etzioni2016ai}.

However, whether the same `practical safety' can be assumed in the case of superintelligence is not obvious. The ability of modern computers to adapt using sophisticated machine learning algorithms makes it even more difficult to make assumptions about the eventual behavior of a superintelligent AI. While computability theory cannot answer this question, it tells us that there are fundamental, mathematical limits to our ability to use one AI to guarantee a null catastrophic risk of another AI (see Figure \ref{fig:halting}).

In closing, it may be appropriate to revisit Norbert Wiener, the founder of the field of \emph{Cybernetics}\cite{Wiener1948-1961}, who compared the literalism of magic to the behavior of computers : 
\begin{quotation}
\emph{More terrible than either of these tales is the fable of the monkey's paw\cite{Jacobs1902}, written by W. W. Jacobs, an English writer of the beginning of the [20th] century. A retired English working-man is sitting at his table with his wife and a friend, a returned British sergeant-major from India. The sergeant-major shows his hosts an amulet in the form of a dried, wizened monkey's paw... [which has] the power of granting three wishes to each of three people... The last [wish of the first owner] was for death... His friend... wishes to test its powers. His first [wish]  is for 200 pounds. Shortly thereafter there is a knock at the door, and an official of the company by which his son is employed enters the room. The father learns that his son has been killed in the machinery, but that the company... wishes to pay the father the sum of 200 pounds... The grief-stricken father makes his second wish -that his son may return- and when there is another knock at the door... something appears... the ghost of the son. The final wish is that the ghost should go away. In these stories the point is that the agencies of magic are literal-minded... The new agencies of the learning machine are also literal-minded. If we program a machine... and ask for victory and do not know what we mean by it, we shall find the ghost knocking at our door.}
\end{quotation}

\section*{Acknowledgments}

We are grateful to Scott Aaronson his insightful comments that helped us contextualize this work. We especially thank our colleague Roberto Moriyon who provided deep insight and expertise that greatly assisted the research. This research was partially supported by the Australian Government as represented by the Department of Broadband, Communications and the Digital Economy and the Australian Research Council through the ICT Centre of Excellence program, by the Spanish MINECO grant TEC2014- 55713-R, by the Regional Government of Madrid (CM) grant Cloud4BigData (S2013/ICE-2894, co-funded by FSE \& FEDER), and by the NSF of China grant 61520106005.

\bibliography{references}
\bibliographystyle{plain}

\end{document}